\documentclass[letterpaper, 10pt, conference]{ieeeconf}
\IEEEoverridecommandlockouts

\usepackage[utf8]{inputenc}
\usepackage[english]{babel}
\usepackage{tabularx}
\usepackage{amsmath,amssymb,xcolor}

\usepackage{amsthm}
\usepackage{enumerate}
\usepackage[noadjust]{cite}

\newcommand{\R}{\ensuremath{\mathbb{R}}}
\newcommand{\N}{\ensuremath{\mathbb{N}}}
\newcommand{\X}{\ensuremath{\mathbb{X}}}
\newcommand{\V}{\ensuremath{\mathbb{V}}}
\newcommand{\Linf}{\ensuremath{\mathcal{L}_\infty}}
\newcommand{\yest}{\ensuremath{\widehat{\kern-1pt\widetilde{y}}\kern1pt}}
\newcommand{\what}{\ensuremath{\widehat{\kern-0.5pt w}}}

\newcommand{\A}{\ensuremath{\mathcal{A}}}

\DeclareMathOperator*{\dom}{dom}

\DeclareMathOperator{\diag}{diag}

\DeclareMathOperator{\interior}{int}

\usepackage{stfloats}

\usepackage{tikz}
\usetikzlibrary{calc}
\makeatletter
\newcommand\currentcoordinate{\the\tikz@lastxsaved,\the\tikz@lastysaved}
\makeatother
\usepackage[utf8]{inputenc}
\usepackage{pgfplots}
\usepackage{pgfgantt}
\usepackage{pdflscape}
\pgfplotsset{compat=newest}
\pgfplotsset{plot coordinates/math parser=false}

\usepackage{hyperref}

\newtheorem{theorem}{Theorem}

\newtheorem{defn}[theorem]{Definition}
\newtheorem{remark}{Remark}
\newtheorem{assumption}{Assumption}
\newtheorem{corollary}[theorem]{Corollary}
\newtheorem{proposition}[theorem]{Proposition}

\let\oldepsilon\epsilon
\let\epsilon\varepsilon
\let\varepsilon\oldepsilon

\let\leq\leqslant
\let\geq\geqslant
\let\hat\widehat
\let\tilde\widetilde
\let\cal\mathcal

\newif\iffull
\fulltrue 

\iffull\else
\fi

\title{\LARGE \bf Event-Triggered State Estimation with Multiple Noisy Sensor Nodes}%
\author{K.J.A. Scheres, M.S.T. Chong, R. Postoyan, W.P.M.H. Heemels
\thanks{Koen Scheres, Michelle Chong and Maurice Heemels are with the Department of Mechanical Engineering, Eindhoven University of Technology, The Netherlands. {\tt\small k.j.a.scheres@tue.nl, m.s.t.chong@tue.nl, m.heemels@tue.nl}}
\thanks{Romain Postoyan is with the Universit\'e de Lorraine, CNRS, CRAN, F-54000 Nancy, France. {\tt\small romain.postoyan@univ-lorraine.fr} His work is supported by the ANR under grant HANDY ANR-18-CE40-0010.}}%

\begin{document}
\maketitle
\thispagestyle{empty}
\pagestyle{empty}
\begin{abstract}
General nonlinear continuous-time systems are considered for which its state is estimated via a packet-based communication network. We assume that the system has multiple sensor nodes, affected by measurement noise, which can transmit at discrete (non-equidistant) points in time. Moreover, each node can transmit asynchronously. For this setup, we develop a state estimation framework, where the transmission instances of the individual sensor nodes can be generated in either time-triggered or event-triggered fashions. In the latter case, we guarantee the absence of Zeno behavior by construction. It is shown that, under the provided design conditions, an input-to-state stability property is obtained for the estimation error with respect to the measurement noise and process disturbances and that the state is thus reconstructed asymptotically in the absence of noise. \iffull A numerical case study shows the strengths of the developed framework.\fi
\end{abstract}
\section{Introduction}
In the last few decades, the quantity and complexity of sensors in Cyber-Physical Systems (CPS) has increased to meet progressing demands on performance, robustness and resilience. Moreover, an increasing number of CPS incorporate (wireless) network technologies to share information across different (physically separated) systems. In particular, the study of Wireless Sensor Networks (WSN) has gained traction in the last two decades. Such WSN have many useful applications, see, e.g., \cite{Jino_Ramson_Moni_2017} and references therein, and are often employed for the state estimation of physically wide spread phenomena.

State estimation has been widely studied and used in the control field, often in the context of stabilizing systems via feedback. Nonetheless, some systems cannot directly actuate (i.e., control) their respective environments or their states only have to be monitored, which motivates considering state estimation as a problem itself. In the context of networked systems, decoupling the control and estimation aspects yields interesting challenges, as we will also see in this paper.

The majority of the existing literature that deals with state estimation over a network often assumes that the underlying system operates in discrete-time, see, e.g., \cite{He_Shin_Xu_Tsourdos_2020,Quevedo_Ahlen_Ostergaard_2010,Sijs_Lazar_2012} and the references therein. Since our physical world is continuous time in nature, it is worthwhile to study the estimation problems using packet-based networks in the continuous-time case as well. Moreover, such a discrete setup is not well-suited to describe phenomena such as asynchronicity between different parts of the system, which is a phenomenon often observed in WSN. Although some results are available for the continuous-time case that deal with multiple sensor modules, see, e.g., \cite{Ren_Al-Saggaf_2018,Wang_Morse_2018}, they often assume that the information can be communicated between different parts of the system continuously. In practice, a packet-based network transmits information at discrete moments. Hence, these works do not apply to the CPS setup.

An additional important aspect, certainly for power-limited devices such as battery-operated sensor modules, is the objective to reduce the frequency of communication required to achieve desirable stability or performance properties. For control purposes, Event-Triggered Control (ETC) has been advocated to this end, see, e.g., \cite{Heemels2012} and the references therein. ETC is a control technique that determines when to exchange information over a (packet-based) network based on well-defined events occurring in the system. It has been shown that these techniques may allow us to greatly reduce the network usage while still guaranteeing similar stability and performance criteria as their time-triggered counterparts.

In the literature on event-triggered state estimation, discrete-time systems are mostly considered, see, e.g., \cite{Shi_Chen_Shi_2014,Muehlebach_Trimpe_2018} and the references therein. Notable exceptions are, e.g., \cite{Zou_Wang_Gao_Liu_2015,Etienne_DiGennaro_Barbot_2017,Niu_Sheng_Gao_Zhou_2020}, where \emph{static} event-triggered state estimation is applied to specific classes of nonlinear systems, and, in the context of \emph{security}, e.g., \cite{Lu_Yang_2017,Liu_Wei_Xie_Yue_2019,Liu_Wang_Yuan_Liu_2019} where linear and Lur'e systems are considered. In all of the above works, measurement noise is ignored. This is not a full surprise, as the presence of measurement noise has been shown to be a challenging issue in ETC \cite{Borgers_Heemels_2014}.

In this context, we study the state estimation problem for perturbed nonlinear systems with multiple sensor nodes, where each node decides when to transmit data, potentially noisy, over the network independently of the other sensors. Our framework covers periodic, time-triggered and dynamic event-triggered packet-based communication behavior. Dynamic event-triggering is used in the sense that each sensor node is equipped with an auxiliary scalar state variable to define the triggering instants. Dynamic triggering strategies are relevant as they typically lead, on average, to larger inter-event times, as shown for ETC in, e.g., \cite{Girard_2015}. When event-triggered communication is used, we prove that the time between two consecutive transmissions is lower bounded by a positive constant, i.e., that no Zeno behavior occurs. Only the knowledge of the upper bound on the size of the noise is needed to potentially reduce the transmission frequency. We guarantee the asymptotic convergence of the state estimates to the true states in the absence of noises. Moreover, when noises are present, we guarantee an input-to-state stability property of the estimation error dynamics, where we show that \emph{space-regularization} (see \cite{Scheres_Postoyan_Heemels_2020}) can be used to obtain more favorable average inter-event times in the presence of non-vanishing noise. \iffull Numerical simulations illustrate the efficacy of our analysis and design framework.\fi

Recently in \cite{Elena2021}, the event-triggered state estimation problem for linear systems and a single sensor node without measurement noise is considered. A key difference in the solution is that we will use local observers in each sensor module, while \cite{Elena2021} uses only first-order dynamics and therefore, requires less computational power in the sensor nodes. However, our method leads to the asymptotic convergence to zero of the estimation error in the absence of disturbances and noises (as opposed to a practical property in \cite{Elena2021}), and applies to a broader range of setups, namely, we consider general nonlinear plant dynamics with multiple noisy sensor nodes.

\section{Preliminaries}
\iffull\subsection{Notation}\fi
The sets of all non-negative and positive integers are denoted $\N$ and $\N_{>0}$, respectively. The field of all reals and all non-negative reals are indicated by $\R$ and $\R_{\geq0}$, respectively. The identity matrix of size $N\times N$ is denoted by $I_N$, and the vectors in $\R^N$ whose elements are all ones or zeros are denoted by $\mathbf{1}_N$ and $\mathbf{0}_N$, respectively.
For any vector $u\in\R^m$, $v\in\R^n$, the stacked vector $\begin{bmatrix}u^\top&v^\top\end{bmatrix}^\top$ is denoted by $(u,v)$. By $\langle\cdot,\cdot\rangle$ and $|\cdot|$ we denote the usual inner product of real vectors and the Euclidean norm, respectively. We denote a matrix $A$ being positive definite (respectively positive semi-definite) by $A\succ0$ ($A\succeq0$) and negative definite (respectively negative semi-definite) by $A\prec0$ ($A\preceq0$). For any $x\in\R^N$, the distance to a closed non-empty set $\cal{A}$ is denoted by $|x|_\cal{A}:=\min_{y\in\cal{A}}|x-y|$. By $\land$ and $\lor$ we denote the logical \emph{and} and \emph{or} operators respectively.

We use the usual definitions for comparison functions $\mathcal{K}$ and $\mathcal{K}_{\infty}$, see \cite{hybridsystems}.

\iffull\subsection{Hybrid systems}\fi
We model hybrid systems using the formalism of \cite{hybridsystems,Cai2009}. \iffull As such, we consider systems $\cal H(F,\cal C,G,\cal D)$ of the form
\begin{equation*}
    \begin{cases}
        \dot{\xi}\in F(\xi,\nu)&(\xi,\nu)\in\cal{C},\\
        \xi^+\in G(\xi,\nu)&(\xi,\nu)\in\cal{D},
    \end{cases}
\end{equation*}
where $\xi\in\R^{n_\xi}$ denotes the state, $\nu\in\R^{n_\nu}$ a disturbance, $\cal{C}\subseteq\R^{n_\xi}\times\R^{n_\nu}$ the flow set, $\cal{D}\subseteq\R^{n_\xi}\times\R^{n_\nu}$ the jump set, $F:\R^{n_\xi}\times\R^{n_\nu}\rightrightarrows\R^{n_\xi}$ the flow map and $G:\R^{n_\xi}\times\R^{n_\nu}\rightrightarrows\R^{n_\xi}$ the jump map, where the maps $F$ and $G$ are possibly set-valued. Loosely speaking, while $(\xi,\nu)\in\cal{C}$, the state can flow continuously according to $\dot{\xi}\in F(\xi,\nu)$. If $(\xi,\nu)\in\cal{D}$, the state can jump as $\xi^+\in G(\xi,\nu)$. If $(\xi,\nu)\in\cal{C}\cap\cal{D}$, the system can either flow or jump. Flow is only allowed if flowing keeps the solution in $\cal C$.
\fi
See \cite{hybridsystems,Cai2009} for more details on the adopted hybrid terminology, such as maximality of solutions, hybrid time domains, the $\Linf$-norm for hybrid signals, etc.

We are interested in systems $\cal H$ that are persistently flowing as defined below.
\begin{defn}
    Given a set of (hybrid) inputs $\mathcal{V}\subseteq\Linf$, a hybrid system $\cal H$ is persistently flowing if all maximal solution pairs $(\phi,\nu)\in\mathcal{S}_\mathcal{H}$ with $\nu\in\mathcal{V}$ are complete in the $t$-direction, i.e., $\sup_t\dom\phi=\sup_t\dom\nu=\infty$.
\end{defn}

\begin{defn}
    \label{def:isps}When $\cal H$ is persistently flowing with $\cal{V}\subseteq\Linf$, we say that a closed, non-empty set $\A\subset\R^{n_\xi}$ is \emph{input-to-state practically stable} (ISpS), if there exist $\gamma\in\cal{K}$, $\beta\in\cal{KL}$ and $d\in\R_{\geq0}$ such that for any solution pair $(\xi,\nu)$ with $\nu\in\cal{V}$
\begin{equation}
    |\xi(t,j)|_\A\leq\beta(|\xi(0,0)|_\A,t)+\gamma(\|\nu\|_\infty)+d,\label{eq:ispsdef}
\end{equation}
for all $(t,j)\in\dom\xi$. If \eqref{eq:ispsdef} holds with $d=0$, then $\cal{A}$ is said to be \emph{input-to-state stable} (ISS) for $\cal{H}$.
\end{defn}
\iffull
To prove that a given non-empty closed set $\cal{A}$ is IS(p)S, we will use the following Lyapunov conditions.
\begin{proposition} \label{prop:lyapunovisps} Consider a persistently flowing system $\cal{H}$ with a set of hybrid inputs $\cal{V}\subseteq\Linf$ and let $\cal{A}\subset\R^{n_\xi}$ be a non-empty closed set. If there exist a continuously differentiable $V:\R^{n_\xi}\to\R_{\geq0}$, $\alpha,\underline{\alpha},\overline{\alpha}\in\cal{K}_\infty$, $\gamma\in\cal{K}$ and $c\in\R_{\geq0}$ such that
\begin{enumerate}\itemsep-0.25em
    \item for any $(\xi,\nu)\in\cal{C}\cup\cal{D}$,
    \vspace{-0.25em}\[\underline{\alpha}(|\xi|_\A)\leq V(\xi)\leq\overline{\alpha}(|\xi|_\A),\]
    \item for all $(\xi,\nu)\in\cal{C}$ and $f\in F(\xi,\nu)$,
    \vspace{-0.25em}\[\left\langle\nabla V(\xi),f\right\rangle\leq-\alpha(|\xi|_\A)+\gamma(|\nu|)+c,\]
    \item for all $(\xi,\nu)\in\cal{D}$ and any $g\in G(\xi,\nu)$,
    \vspace{-0.4em}\[V(g)-V(\xi)\leq0,\]
\end{enumerate}
then $\cal{A}$ is ISpS, and it is ISS if $c=0$.
\end{proposition}
\begin{proof}[Sketch of proof]
Let $(\xi,w)$ be a solution to $\cal H$, $(t,j)\in\dom \xi$ and $0=t_{0}\leq t_{1} \leq \ldots \leq t_{j+1}=t$ satisfy
\begin{equation}
\dom \xi \cap ([0,t]\times\{0,\ldots,j\}) = {\underset{i\in\{0,\ldots,j\}}\bigcup}[t_{i},t_{i+1}]\times\{i\}.
\end{equation}
For each $i\in\{1,2,\ldots,j\}$ and for almost all $s\in[t_{i},t_{i+1}]$, item 2) of Proposition 2 implies that $\left\langle\nabla V(\xi(s,i)),\dot\xi(s,i)\right\rangle\leq-\alpha(|\xi(s,i)|_\cal{A})+\gamma(|w(s)|)+c$.
We can then invoke similar arguments as in \cite[Lemma 2.14]{Sontag1995} to obtain the desired result as: (i) $V$ does not increase at jumps according to item 3) of Proposition 2, (ii) item 1) holds, and (iii) $\cal{H}$ is persistently flowing.
\end{proof}
\fi

\section{Problem formulation}\label{sec:problemformulation}
\begin{figure}[b]
    \tikzstyle{block} = [draw, fill=gray!30, rectangle, minimum height=1cm, minimum width=1cm]
    \tikzstyle{etm} = [draw, fill=white, rectangle, minimum height=0.75cm, minimum width=0.75cm]
    \tikzstyle{network} = [draw=none, rectangle, minimum height=0.75cm, minimum width=3.75cm]
    \tikzstyle{dots} = [draw=none,text centered,text height=2.1ex]
    \centering
    \begin{tikzpicture}[node distance=1.5cm, >=latex, line width=1pt, color=black,Pattern/.style={pattern=north east hatch, pattern color=blue!20, hatch distance=15pt, hatch thickness=6pt}]
        \node[block] at (0,0) (system) {$\Sigma$};
        \node[draw=none, left of=system, node distance=1.25cm] (v) {$v$};
        \node[dots,right of=system,node distance=1.75cm] (dots) {$\vdots$};
        \node[draw, circle, minimum size=2mm, inner sep=0pt, above of=system, node distance=1.25cm] (sum) {$+$};
        \node[draw=none, left of=sum, node distance=1.125cm] (w) {$w_1$};
        \node[draw, circle, minimum size=2mm, inner sep=0pt, below of=system, node distance=1.25cm] (sumN) {$+$};
        \node[draw=none, left of=sumN, node distance=1.125cm] (wN) {$w_N$};
        \node[etm, right of=sum,node distance=1.75cm,text width=1.25cm,align=center, label={90:{\footnotesize Sensor 1}}] (TX) {TX/RX\\$\Omega_1$};
        \node[etm, right of=sumN,node distance=1.75cm,text width=1.25cm,align=center, label={90:{\footnotesize Sensor N}}] (TXN) {TX/RX\\$\Omega_N$};
        \node[network,right of=system, rotate=90, node distance=3.75cm] (net) {Network};

        \draw[dashed](net.north west)--(net.north east) (net.south west)--(net.south east);
        \draw[->] (v)--(system);
        \draw[->] (system)--(sum) node[midway, left] {$y_1$};
        \draw[->] (w)--(sum);
        \draw[->] (sum)--(TX) node[midway, above] {$\widetilde{y}_1$};
        \draw[->] (system)--(sumN) node[midway, left] {$y_N$};
        \draw[->] (wN)--(sumN);
        \draw[->] (sumN)--(TXN) node[midway, below] {$\widetilde{y}_N$};
        \draw[->, dashdotted] (TX.15)--(net.north|-TX.15) node[midway, above] {$\yest_1$};
        \draw[<-, dashdotted] (TX.345)--(net.north|-TX.345) node[midway, below] {$\yest$};
        \draw[->, dashdotted] (TXN.15)--(net.north|-TXN.15) node[midway, above] {$\yest_N$};
        \draw[<-, dashdotted] (TXN.345)--(net.north|-TXN.345) node[midway, below] {$\yest$};

        \node[etm, right of=net, node distance=1.5cm] (RX) {RX};
        \node[block, right of=RX, node distance=1.5cm, minimum height=1cm, fill=blue!20] (observer) {$\Omega_0$};
        \draw[->, dashdotted] (net.south|-RX)--(RX) node[midway, above] {$\yest$};
        \draw[->] (RX)--(observer);
        \matrix [draw,below left, color=black!80] at (current bounding box.north east) {
            \draw[->] (\currentcoordinate)--++(0.75,0) node[near end, right] {\footnotesize\quad continuous};\\
            \draw[->, dashdotted] (\currentcoordinate)--++(0.75,0) node[near end, right] {\footnotesize\quad discrete};\\
        };
    \end{tikzpicture}
    \caption{\footnotesize Event-triggered state estimation setup over a packet-based network. $\Sigma$ is the to-be-observed system, $\Omega_0$ the remote observer. In case the transmissions are event-based, each sensor runs a local observer $\Omega_i$, $i\in\mathcal{N}$, and the sensor measurements are broadcast over the network. Dashed arrows indicate packet-based (discrete) communication. Solid arrows denote continuous communication. Due to Assumption \ref{ass:init}, the indices for $\yest$ have been omitted.}
    \label{fig:setup}
\end{figure}
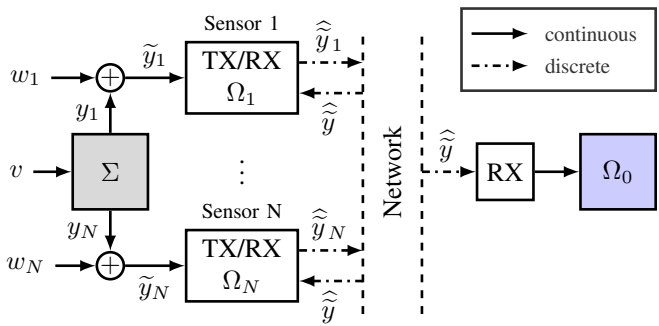
We address the event-triggered state estimation problem of a dynamical system with multiple asynchronous (TX/RX) sensor nodes, where the measurements are transmitted over a packet-based communication network, as shown in Fig. \ref{fig:setup}. As we will see, due to our particular setup and design, time-triggered and periodic communication will be special cases of the event-triggered setting, which will be treated in detail.

Next, we explain the role of each component in the setup as illustrated in Fig. \ref{fig:setup} below. First of all, the system $\Sigma$ has dynamics
\begin{equation}\label{eq:systemdynamics}
    \Sigma:\begin{cases}\dot{x}=f_p(x,v)\\y_i=h_{p,i}(x),\;\widetilde{y}_i=y_i+w_i,&i\in\mathcal{N}\end{cases}
\end{equation}
with state $x\in\R^n$, $N\in\N_{>0}$ sensor nodes measuring $\widetilde{y}_i\in\R^{m_i}$, $i\in\mathcal{N}:=\{1,2,\ldots,N\}$, system disturbance $v$ taking values in $\R^p$, a map $f_p:\R^n\times\R^p\to\R^n$, continuously differentiable output maps $h_{p,i}:\R^n\to\R^{m_i}$ and measurement noise $w_i$ taking values in $\R^{m_i}$, $i\in\mathcal{N}$. We assume that the disturbance signal $v$ and measurement noises $w_i$, $i\in\mathcal{N}$, are in $\Linf$. We also assume that \eqref{eq:systemdynamics} is forward complete, i.e., for all $x(0)\in\R^n$ and all $v\in\Linf$, there exists a unique solution $x$ to \eqref{eq:systemdynamics} for all $t\in\R_{\geq0}$.

The objective of the to-be-designed observer $\Omega_0$ in Fig. \ref{fig:setup} is to asymptotically reconstruct the state $x$ of system $\Sigma$. Due to the presence of the packet-based network, the observer $\Omega_0$ does not have (continuous) access to the outputs $\widetilde{y}_i$, but only to the estimates gathered in $\yest^0:=(\kern1pt\yest_1^0,\yest_2^0,\ldots,\yest_N^0)$. Each sensor node $i$ has a local triggering condition, which determines when $\widetilde{y}_i$ needs to be communicated over the network. These transmissions occur at times denoted by $t^i_k$, $k\in\N$, at which the estimate of the $i$-th output in $\Omega_0$ is updated according to
\begin{equation}
    \yest_i^0((t^i_k)^+)=\widetilde{y}_i(t^i_k),\quad k\in\N.\label{eq:update}
\end{equation}
We assume that there are no transmission delays. In the future, we plan on taking transmission delays into account, possibly inspired by \cite{Heemels_Teel_Wouw_Nesic_2010,Dolk2017}.

For each sensor node $i$, the transmission times $t_k^i$, $k\in\N$, are determined by the corresponding $i$-th trigger mechanism. In our description, we will focus on event-triggered mechanisms, but, as stated before, the framework that we present includes (non-uniform) time-based sampling as well as a special case. In the Event-Triggered Mechanism (ETM), the locally available output information is used to determine the transmission instances $t_k^i$, $k\in\N$. For such ETMs, it is essential to guarantee that the time between consecutive transmissions of trigger mechanism $i$ is \emph{lower bounded} by a positive minimum inter-event time (MIET) $\tau_\text{MIET}^i\in\R_{>0}$, i.e., $0<\tau_\text{MIET}^i\leq t_{k+1}^i-t_k^i$, for any $k\in\N$.

The objective of this paper can now be formalized as the joint design of a) the remote observer $\Omega_0$, that only uses the information $\widetilde{y}_i$ received at times $t_k^i$, $k\in\N$, $i\in\mathcal{N}$, and b) the event- or time-triggered transmission mechanisms in the sensor nodes that generate the times $t_k^i$, $k\in\N$, such that
\begin{enumerate}[(i)]
    \item the observer $\Omega_0$ reconstructs the state (approximately) in the presence of process disturbances $v$ and measurement noise $w_i$,
    \item a positive lower bound on the inter-event times (MIET) is guaranteed for each node and communication resources are saved by transmitting sporadically.
\end{enumerate}
\section{Proposed solution structure}\label{sec:solutionstructure}
We proceed by emulation, i.e., we assume that the observer is designed in continuous-time, ignoring the network, to ensure robust global asymptotic stability of the origin for the corresponding estimation error system, as formalized in the sequel. Any existing observer designs suitable for plant \eqref{eq:systemdynamics} can be used for this task. The remote observer $\Omega_0$ receives sensor measurements $\widetilde{y}:=(\widetilde{y}_1,\widetilde{y}_2,\ldots,\widetilde{y}_N)$ over the network and generates the state estimate $\chi_0$. The observer dynamics is given by
\begin{equation}\label{eq:observerdynamics}
    \Omega_0:\dot{z}_0=f_o(z_0,\yest^0),\;\chi_0=h_o(z_0),
\end{equation}
with the observer state $z_0\in\R^{q}$, $q\in\N_{>0}$, the state estimate of the system $\chi_0\in\R^n$, $f_o:\R^q\times\R^m\to\R^q$ with $m:=\sum_{i\in\mathcal{N}}m_i$ and $h_o:\R^q\to\R^n$, and where $\yest^0:=(\kern1pt\yest_1^0,\yest_2^0,\ldots,\yest_N^0)$ is the estimate of $\widetilde{y}:=(\widetilde{y}_1,\widetilde{y}_2,\ldots,\widetilde{y}_N)$. Since $z_0$ admits a different number of states than the system $\Sigma$, Kalman-like observers and immersion-based observers, see, e.g., \cite{Ticlea_Besancon_2007}, may be considered.

In between the transmission times $t_k^i$, $k\in\N$, the estimate $\yest^0_i$, $i\in\mathcal{N}$, evolves according to the \emph{holding function} $f_{h,i}:\R^q\times\R^m\to\R^{m_i}$, which is part of the design, given by the differential equation
\begin{equation}
    \dot{\kern5pt\yest^0_i}=f_{h,i}(z_0,\yest^0).\label{eq:hf}
\end{equation}
Note that only locally available information is used in the holding function and the expression of $f_{h,i}$ is specified in the following.

Each $i$-th sensor has a \emph{trigger} that determines when the value of $\widetilde{y}_i$ is transmitted over the network. All $i\in\mathcal{N}$ sensors are equipped with a transmission and receiver module (TX and RX) and a local observer $\Omega_i$. At times $t_k^i$, generated by trigger $i$, the value of $\widetilde{y}_i$ is broadcast over the network and consequently received by all $j$ sensor modules, $j\in\mathcal{N}$, and the remote observer $\Omega_0$. The local observer $\Omega_i$, $i\in\mathcal{N}$, is chosen to be a copy of the remote observer $\Omega_0$, i.e., its dynamics are taken as $\dot{z}_i=f_o(z_i,\yest^i)$ and $\chi_i=h_o(z_i)$ and the holding functions as $\frac{d}{dt}\yest^i_j=f_{h,j}(z_i,\yest^i)$ for all $j\in\mathcal{N}$.

In the event-triggered setting, the (local) observers play a role in generating the transmission events, as we will make more precise below. In the time-triggered or periodic case, these local observers are omitted and replaced by a local clock to measure the time elapsed since the last transmission, in which case the receiver module can be omitted from all sensors.

We adopt the following assumption.
\begin{assumption}\label{ass:init}
    At the initial time $t=0$, $z_i(0)=z_j(0)$ and $\yest^i(0)=\yest^j(0)$ for all $i,j\in\mathcal{N}\cup\{0\}$.
\end{assumption}
This assumption essentially entails that at time $t=0$, the states $z_i$ and estimates of the outputs $\yest^i$ are the same across all the sensors (local observers $\Omega_i$, $i\in\mathcal{N}$) and the remote observer $\Omega_0$. In practice, this can be achieved by using a boot-up sequence or initialization procedure. Since we do not restrict ourselves to any particular initial condition, a straightforward choice would be to set $z_i=\mathbf{0}_q$ and $\yest^i=\mathbf{0}_m$ for all $i\in\mathcal{N}\cup\{0\}$ at $t=0$. In future work, we aim to relax Assumption \ref{ass:init}.

Due to Assumption \ref{ass:init} and \eqref{eq:update} and \eqref{eq:hf}, we have that $z_i(t)=z_j(t)$ and $\yest^i(t)=\yest^j(t)$ for all $t\in\R_{\geq0}$ and any $i,j\in\mathcal{N}\cup\{0\}$. Henceforth we will drop the indices for $z_i$ and $\yest^i$ and we will denote the observer state and output estimate as $z$ and $\yest$, respectively.

The holding function $f_{h,i}$ in \eqref{eq:hf} is designed as
\begin{equation}\label{eq:hfdesign}
    f_{h,i}(z,\yest)=\frac{\partial h_{p,i}(h_o(z))}{\partial \chi}f_p(h_o(z),0).
\end{equation}
The choice for the particular structure of this holding function is explained in Remark \ref{rem:holdingfunction} below.

In the event-triggered setting, we allow the trigger mechanism to depend on a local auxiliary variable, denoted by $\eta_i\in\R_{\geq0}$, $i\in\cal N$, as in dynamic triggering of e.g. \cite{Girard_2015,Dolk2017}. The local variable $\eta_i$ has flow dynamics $\dot{\eta}_i=\Psi_i(o_i)$ and jump dynamics $\eta_i^+=\eta_i^0(o_i)$, where $o_i\in\R^{o_i}$ denotes the locally available information and the functions $\Psi_i:\R^{o_i}\to\R$ and $\eta_i^0:\R^{o_i}\to\R$ will be designed appropriately.
Then, the triggering instances $t_k^i$, $i\in\mathcal{N}$, are determined by, for all $k\in\N$,
\begin{equation}\label{eq:triggertimes}
    t_0^i=0,\;t_{k+1}^i:=\inf\{t\geq t_k^i+\tau_\text{MIET}^i\mid\eta_i(t)\leq0\}.
\end{equation}
Additionally, each sensor $i\in\mathcal{N}$ has a local timer to keep track of the time elapsed since the last transmission of sensor $i$, denoted $\tau_i\in\R_{\geq0}$, whose dynamics are given by, for all $i\in\mathcal{N}$,
\begin{equation}\label{eq:timers}
    \dot\tau_i=1,\; \tau_i^+=0,
\end{equation}
where resets take place when node $i$ transmits $\widetilde{y}_i$. Note that the local clocks do not have to be synchronized. In solving the problem formulated at the end of Section \ref{sec:problemformulation}, we have to appropriately select $\Psi_i$, $\eta_i^0$ and $\tau_\text{MIET}^i$, $i\in\mathcal{N}$, which will be done in the following sections after providing a suitable hybrid model.
\section{Hybrid Model}\label{sec:hybridmodel}
To proceed with the analysis, we model the overall system as a hybrid system using the formalism of \cite{hybridsystems,Cai2009} for which a jump corresponds to a transmission of one of the outputs $\widetilde{y}_i$ over the network.

We define for this purpose, $\widehat{y}_i$ and $\what_i$, where
\begin{equation}
    \begin{aligned}
        \widehat{y}_i((t^i_k)^+)&=y_i(t^i_k),&\dot{\kern-1.2pt\widehat{y}}_i&=f_{h,i}(z,\yest),\\
        \what_i((t^i_k)^+)&=w_i(t^i_k),&\dot{\what}_i&=0.
    \end{aligned}\label{eq:holdingfunction}
\end{equation}
Here, $\what_i$ is the sampled measurement noise and $\widehat{y}_i$ the estimated noise-free output. From \eqref{eq:update} and \eqref{eq:hf}, we deduce that
\begin{equation}
    \yest_i=\widehat{y}_i+\what_i\label{eq:yestsplit}.
\end{equation}
Let the \emph{noise-free} network-induced error of the $i$-th sensor, $i\in\mathcal{N}$, be denoted by
\begin{equation}
    \epsilon_i:=\widehat{y}_i-y_i.\label{eq:networkerror}
\end{equation}
The \emph{measured} network-induced error (influenced by the measurement noise), is denoted by
\begin{equation}
    \widetilde{\epsilon}_i:=\yest\kern-1pt_i-\widetilde{y}_i=\epsilon_i+\what_i-w_i.\label{eq:measnetworkerror}
\end{equation}
The estimation error is denoted $e:=\chi-x$, whose dynamics between two successive transmission instances is given by
\begin{equation}\label{eq:dote}
    \dot{e}=\dot{\chi}-\dot{x}=\frac{\partial h_o(z)}{\partial z}f_o(z,\yest)-f_p(x,v).
\end{equation}

The $i$-th noise-free network-induced error $\epsilon_i$, $i\in\mathcal{N}$, has dynamics
\begin{equation}\label{eq:dotepsilon}
    \dot{\epsilon}_i=\dot{\kern-1pt\widehat{y}}_i-\dot{y}_i=f_{h,i}(z,\yest)-\frac{\partial h_{p,i}(x)}{\partial x}f_p(x,v).
\end{equation}

\begin{remark}\label{rem:holdingfunction}
    An appropriate choice for the holding function is important to reduce communication. By picking $f_{h,i}(z,\yest)$ as in \eqref{eq:hfdesign}, we have for $\chi=h_o(z)=x$, and thus $e=0$, and in the absence of disturbances, that $\dot{\epsilon}_i=0$, i.e., there is no ``need'' to communicate over the network. In contrast, if we employ, e.g., a zero-order hold, i.e., $f_{h,i}(z)=0$ in \eqref{eq:holdingfunction}, this is not the case \emph{even in the absence of disturbances}, since $f_{p}(x,0)$ may be non-zero, i.e., the network-induced error may grow over time even though, initially, the observer has a correct estimate of the state. Since the holding function depends on the observer state, a local `copy' of the observer is needed at every sensor for this implementation.
\end{remark}

The state for the hybrid model $\mathcal{H}$ is $\xi:=(z,e,\epsilon,\what,\tau,\eta)\in\X$, where $\tau:=(\tau_1,\tau_2,\ldots,\tau_N)$, $\eta:=(\eta_1,\eta_2,\ldots,\eta_N)$ and $\X:=\R^q\times\R^n\times\R^m\times\R^m\times\R^N_{\geq0}\times\R^N_{\geq0}$. Additionally, we concatenate the disturbances and the noises as $\nu:=(v,w)\in\V$ with $\V:=\R^n\times\R^m$.

Using \eqref{eq:yestsplit} and \eqref{eq:measnetworkerror}, we can rewrite $\widehat{\widetilde{y}}$ as
\begin{equation}\label{eq:yest}
    \yest=y+\epsilon+\what
\end{equation}
with $y:=(y_1,y_2,\ldots,y_N)$, $\epsilon:=(\epsilon_1,\epsilon_2,\ldots,\epsilon_N)$ and $\what:=(\what_1,\what_2,\ldots,\what_N)$. Based on \eqref{eq:holdingfunction}, \eqref{eq:dote} and \eqref{eq:dotepsilon}, the flow map $F:\X\times\V\to\X$ is given by
\begin{equation}\label{eq:ttflowmap}
    \begin{gathered}
        F(\xi,\nu):=\big(f_z(z,e,\epsilon,\what),f_e(z,e,\epsilon,\what,v),\hspace{3em}\\
            \hspace{7em} g(z,e,\epsilon,\what,v), \mathbf{0}_N, \mathbf{1}_N, \Psi(o)\big).
    \end{gathered}
\end{equation}
where $\Psi(o):=(\Psi_1(o_1),\Psi_2(o_2),\ldots,\Psi_N(o_N))$ is the trigger dynamics to-be-designed in Section \ref{sec:design}. From \eqref{eq:observerdynamics} and \eqref{eq:yest} we obtain
\begin{equation}
    f_z(z,e,\epsilon,\what):=f_o(z,{h_p}(h_o(z)-e)+\epsilon+\what),
\end{equation}
where $h_p(x):=\left(h_{p,1}(x),h_{p,2}(x),\ldots,h_{p,N}(x)\right)$. From \eqref{eq:dote} we obtain
\begin{equation}
    \begin{gathered}
        f_e(z,e,\epsilon,\what,v):=\frac{\partial h_o(z)}{\partial z}f_z(z,e,\epsilon,\what)-f_p(h_o(z)-e,v).
    \end{gathered}
\end{equation}
From \eqref{eq:dotepsilon} and \eqref{eq:hfdesign} we obtain
\begin{equation}
    g(z,e,\epsilon,\what,v):=\frac{\partial h_{p}(x)}{\partial x}\left(f_p(h_o(z),0)-f_p(x,v)\right).
\end{equation}
The flow set $\mathcal{C}\subseteq\X$ is given by $\mathcal{C}:=\bigcap_{i\in\mathcal{N}}\mathcal{C}_i$, with
\begin{equation}\label{eq:flowseti}
    \mathcal{C}_i:=\{\xi\in\X\mid\eta_i\geq0\}.
\end{equation}
The jump set $\mathcal{D}\subseteq\X$ is given by $\mathcal{D}:=\bigcup_{i\in\mathcal{N}}\mathcal{D}_i$, with
\begin{equation}
    \mathcal{D}_i:=\left\{\xi\in\X\mid\tau_i\geq\underline{\tau}_i\right\},\label{eq:ttjumpset}
\end{equation}
where $\underline{\tau}_i\in(0,\tau_\text{MIET}^i]$, $i\in\mathcal{N}$, are positive (but arbitrarily small) constants and $\tau_\text{MIET}^i>0$ is to-be-designed. Since we will ensure that $\eta_i\geq0$ for $\tau_i\in[0,\tau_\text{MIET}^i]$ by appropriate design of the trigger dynamics $\dot{\eta}_i=\Psi_i(o_i)$, taking the sets \eqref{eq:flowseti} and \eqref{eq:ttjumpset} admits both the time-triggered and event-triggered communication cases as particular solutions, and, hence, stability and persistently flowing properties can be proven for both settings through one model. Nonetheless, the existence and completeness of solutions has to be treated separately for both cases. The sets $\mathcal{C}_i$ and $\mathcal{D}_i$ can be redesigned appropriately to allow only one or the other, e.g.,
\begin{itemize}\setlength\itemsep{0pt}
    \item time-triggered communication:
    \begin{equation}\label{eq:ttc}
        \begin{cases}
           \mathcal{C}_i:=\{\xi\in\X\mid\tau_i\in[0,\tau_\text{MIET}^i]\}\\
           \mathcal{D}_i:=\{\xi\in\X\mid\tau_i\geq\underline{\tau}_i\},
        \end{cases}
    \end{equation}
    \item event-triggered communication:
    \begin{equation}\label{eq:etc}
        \begin{cases}
           \mathcal{C}_i:=\{\xi\in\X\mid\eta_i\geq0\}\\
           \mathcal{D}_i:=\{\xi\in\X\mid\eta_i\leq0\land\tau_i\geq\tau_\text{MIET}^i\}
        \end{cases}
    \end{equation}
\end{itemize}

Based on \eqref{eq:holdingfunction} and \eqref{eq:networkerror} the jump map $G:\X\times\V\to\X$ is given by $G(\xi,\nu):=\bigcup_{i\in\mathcal{N}}G_i(\xi,\nu)$, where
\begin{align}
    &G_i(\xi,\nu):=\begin{cases}
        \left\{\overline{G}_i(\xi,\nu)\right\},&\text{if }\xi\in\mathcal{D}_i,\\
        \emptyset,& \text{if }\xi\not\in\mathcal{D}_i,
    \end{cases}\\
    &\overline{G}_i(\xi,\nu):=(z,e,\overline{R}_i\epsilon,\overline{R}_i\what+R_iw,\overline{S}_i\tau,\overline{S}_i\eta+S_i\eta^0(o)),\nonumber
\end{align}
and $\eta^0(o):=(\eta_1^0(o_1),\eta_2^0(o_2),\ldots,\eta_N^0(o_N))$ is the (to be designed) jump map of $\eta$, $R_i:=\diag(\Delta_{i,1},\Delta_{i,2},\ldots,\allowbreak\Delta_{i,N})$ is a block-diagonal matrix with $\Delta_{i,j}=0_{m_j}$ if $i\neq j$ and $\Delta_{i,j}=I_{m_j}$ if $i=j$ where $0_k$ is the $k\times k$ zero matrix, $\overline{R}_i=I-R_i$, $S_i$ is a diagonal matrix with the $ii$-th element 1 and all other elements 0 and $\overline{S}_i=I_N-S_i$. Note that jumps are allowed for all $\tau_i\geq\underline{\tau}_i$. Hence, this model captures solutions pertaining to periodic, time-triggered and event-triggered communication behavior.

The first item of the problem formulation at the end of Section \ref{sec:problemformulation} can now be formally stated as designing the constants $\tau_\text{MIET}^i>0$ and functions $\Psi_i(o_i)$ and $\eta^0_i(o_i)$ such that the set $\cal A:=\{\xi\in\X\mid e=0\land\epsilon=0\land\eta=0\}$ is IS(p)S in the sense of Definition \ref{def:isps}.
\section{Design and analysis}\label{sec:design}
We denote the hybrid model as defined in Section \ref{sec:hybridmodel} by $\mathcal{H}$. Inspired by \cite{Dolk2017,Carnevale_Teel_Nesic_2007}, in this section, conditions are presented such that the ETM given by \eqref{eq:ttc} or \eqref{eq:etc} ensures IS(p)S for the estimation error system \eqref{eq:dote}. The following assumption ensures that the observer dynamics does not have finite escape times.
\begin{assumption}\label{ass:zboundedness}
    The observer \eqref{eq:observerdynamics} is forward complete such that for all $z(0)\in\R^q$ and all measurable signals $\yest$, there exists a unique solution $z$ to \eqref{eq:observerdynamics} for all $t\in\R_{\geq0}$.\footnote{In the case of unstable plants \eqref{eq:systemdynamics}, we may obtain the case where $\yest\notin\Linf$. However, this is not an issue as $\hat{\tilde{y}}$ is typically used in the observer $\Omega_{i}$ via output injection, e.g. $h_p(\chi)-\yest$.}
\end{assumption}
Additionally, we require that the dynamics of the network-induced error system satisfies the next property.
\begin{assumption}\label{ass:W}
    For all $i\in\mathcal{N}$, there exist locally Lipschitz functions $W_i:\R^{m_i}\to\R_{\geq0}$ and constants $\underline{\beta}_{i}, \overline{\beta}_{i} > 0$ such that, for all $\epsilon_i\in\R^{m_i}$,
    \begin{equation}\label{eq:Wsandwich}
        \underline{\beta}_i|\epsilon_i|\leq W_i(\epsilon_i)\leq \overline{\beta}_i|\epsilon_i|.
    \end{equation}
    Additionally, there exist constants $L_i\geq0$ and a continuous function $H_i:\R^q\times\R^n\times\R^m\times\R^m\times\R^p\to\R_{\geq0}$ such that for almost all $\epsilon_i$, and all $v\in\R^p$, $z\in\R^q$, $e\in\R^n$ and $\what\in\R^m$,
    \begin{equation}\label{eq:Wder}
        \left\langle\frac{\partial W_i(\epsilon_i)}{\partial \epsilon_i},g(z,e,\epsilon,\what,v)\right\rangle\leq L_iW_i(\epsilon_i)+H_i(z,e,\what,v).
    \end{equation}
\end{assumption}
The inequality in \eqref{eq:Wder} is loosely speaking an upper bound on the growth of $\epsilon_i$ between successive transmission instants, i.e., that $\epsilon_i$ grows at most exponentially. This condition is always satisfied, e.g., linear systems or when the gradient of $W_i$ is bounded (almost everywhere) by a given constant and $g(z,e,\epsilon,\what,v)$ satisfies a linear growth condition. In the following, we will omit the arguments of $H_i$, as specified in Assumption \ref{ass:W}, for brevity.
\begin{assumption}\label{ass:V}
    There exist a locally Lipschitz function $V:\R^n\to\R_{\geq0}$, $\mathcal{K}_\infty$-functions $\alpha,\underline{\alpha},\overline{\alpha},\zeta$ and, for all $i\in\mathcal{N}$, positive semi-definite functions $\varrho_i:\R^{m_i}\to\R_{\geq0}$ and constants $\mu_i,\gamma_i>0$ such that for all $e\in\R^n$
    \begin{equation}\label{eq:Vsandwich}
        \underline{\alpha}(|e|)\leq V(e) \leq \overline{\alpha}(|e|),
    \end{equation}
    and for almost all $e\in\R^n$ and all $z\in\R^q$, $\nu\in\V$ and $\epsilon,\what\in\R^m$
    \begin{equation}\label{eq:Vder}
        \begin{aligned}
            &\left\langle \nabla V(e),f_e(z,e,\epsilon,\what,v)\right\rangle\leq-\alpha(|e|)+\zeta(|\nu|)\\
            &+\sum_{i\in\mathcal{N}}\left(-\varrho_i(|q_i(z,e,w)|)-H_i^2+(\gamma_i^2-\mu_i)W_i^2(\epsilon_i)\right),
        \end{aligned}
    \end{equation}
    where $q_i(z,e,w):=h_{p,i}(\chi)-\widetilde{y}_i=h_{p,i}(h_o(z))-h_{p,i}(h_o(z)-e)-w_i$ is the output estimation error.
\end{assumption}
Assumption \ref{ass:V} is an $\mathcal{L}_2$-gain condition from $W_i$ to $H_i$. \iffull As we will show in Section \ref{sec:numericalexample}, in the case of observable linear systems, the conditions of Assumption \ref{ass:V} can be met by solving a linear matrix inequality (LMI).\fi

Consider now for all $i\in\mathcal{N}$, the functions $\phi_i:\R_{\geq0}\to\R_{\geq0}$, which evolve according to
\begin{equation}\label{eq:phidot}
    \frac{d}{d\tau_i}\phi_i\in(\omega_i(\tau_i)-1)\left(2L_i\phi_i+\gamma_i(\phi^2_i+1)\right),
\end{equation}
where the initial conditions $\phi_i(0)$ are to be specified and where $\omega_i:\R_{\geq0}\rightrightarrows[0,1]$ is defined as
\begin{equation}\label{eq:omega}
    \omega_i(\tau_i):=\begin{cases}
        \{0\},&\text{if}\;\tau_i\in[0,\tau_\text{MIET}^i),\\
        [0,1],&\text{if}\;\tau_i=\tau_\text{MIET}^i,\\
        \{1\},&\text{if}\;\tau_i>\tau_\text{MIET}^i.
    \end{cases}
\end{equation}
The constant $\tau_\text{MIET}^i>0$ is determined by the solution of the differential equation $\dot{\bar\phi}_i=-2L_i\bar{\phi}_i-\gamma_i(\bar{\phi}_i^2+1)$ with $\bar\phi_i(0)=\lambda_i^{-1}$, where $\lambda_i\in(0,1)$ is a tuning parameter. Indeed, $\tau^i_\text{MIET}$ is obtained as the value for which $\bar\phi_i(\tau^i_\text{MIET})=\lambda_i$, i.e.,
\begin{equation}\label{eq:taumiet}
    \tau^i_\text{MIET}:=\begin{cases}
        \frac{1}{L_ir_i}\arctan\left(\frac{r_i(1-\lambda_i)}{2\frac{\lambda_i}{1+\lambda_i}\left(\frac{\gamma_i}{L_i}-1\right)+1+\lambda_i}\right), & \gamma_i>L_i, \\
        \frac{1}{L_i}\frac{1-\lambda_i}{1+\lambda_i}, & \gamma_i=L_i, \\
        \frac{1}{L_ir_i}\operatorname{arctanh}\left(\frac{r_i(1-\lambda_i)}{2\frac{\lambda_i}{1+\lambda_i}\left(\frac{\gamma_i}{L_i}-1\right)+1+\lambda_i}\right), &   \gamma_i<L_i,
    \end{cases}
\end{equation}
where $r_i:=\sqrt{\left|\big(\tfrac{\gamma_i}{L_i}\big)^2-1\right|}$.

\iffull
\begin{proposition}\label{prop:phi}
    For all $i\in\mathcal{N}$, $\phi_i(\tau_i)\in[\lambda_i,\lambda_i^{-1}]$ for all $\tau_i\in\R_{\geq0}$. Moreover, $\phi_i(\tau_i)=\lambda_i$ for all $\tau_i\geq\tau_\text{\normalfont MIET}^i$.
\end{proposition}
\begin{proof}[Sketch of proof]
    Consider the discontinuous function
    $\bar\omega_i(\tau_i)=0$ if $\tau_i\in[0,\tau_\text{MIET}^i]$ and $\bar\omega_i(\tau_i)=1$ if $\tau_i>\tau_\text{MIET}^i$. From \cite{Carnevale_Teel_Nesic_2007}, we know that for this (discontinuous) differential equation, $\phi_i\in[\lambda_i,\lambda^{-1}_i]$. From \cite[Theorem 4.3]{hybridsystems}, we know that taking the solution to the differential inclusion where we take the convex closure of $\bar\omega_i$ produces the same solution.
\end{proof}
\fi
The function $\omega_i$, $i\in\mathcal{N}$, in \eqref{eq:omega} is defined such that the flow map $F$ is outer semi-continuous to ensure that the hybrid system $\mathcal{H}$ satisfies the hybrid basic conditions as presented in \cite[Assumption 6.5]{hybridsystems}.

We are now ready to state the main result of this paper\iffull, whose sketch of proof is given in the appendix\fi.
\begin{theorem}
    Consider the hybrid system $\mathcal{H}$ and suppose Assumptions 1-4 hold. We define for all $i\in\cal{N}$, $\xi\in\X$ and $\nu\in\V$, $\Psi_i$ as
    \begin{equation}
        \Psi_i(o_i):=\varrho_i(|q_i(z,e,w)|)-\omega_i(\tau_i)\overline{\gamma}_i\beta_iW_i^2(\widetilde{\epsilon}_i)-\sigma_i(\eta_i)+s_i\label{eq:trigger},
    \end{equation}
    where $\overline{\gamma}_i=2\gamma_i\lambda_iL_i+\gamma_i^2(1+\lambda_i^2)$, $\beta_i=2\frac{\overline{\beta}^2_i}{\underline{\beta}^2_i}$, $\sigma_i$ is any $\mathcal{K}_\infty$-function, $s_i\geq0$ is a tuning parameter and $\varrho_i,q_i,L_i,\gamma_i,\lambda_i$ and $W_i$ come from Assumptions \ref{ass:W} and \ref{ass:V}. Additionally, we design the trigger resets $\eta_i^0$ as
    \begin{equation}\label{eq:reset}
        \eta_i^0:=0.
    \end{equation}
    The system $\mathcal{H}$ with trigger dynamics \eqref{eq:trigger} and reset \eqref{eq:reset} is persistently flowing and the set $\cal A:=\{\xi: e=0\land\epsilon=0\land\eta=0\}$ is ISS if $s_i=0$ for all $i\in\mathcal{N}$. It is ISpS if, for any $i\in\mathcal{N}$, $s_i>0$.\label{thm:trigger}
\end{theorem}

If an upper-bound $\overline{w}_i$ on the $\Linf$-norm of the measurement noise $w_i$ is explicitly known for node $i\in\mathcal{N}$, as defined in the following assumption, we may improve the average inter-event times, as will be explained in the sequel.
\begin{assumption}\label{ass:noise}
    For each $i\in\cal{N}$, $w_i\in\Linf$ with $w_i(t,j)$ taking values in $\cal{W}_i$ for all $(t,j)\in\dom w_i$, where $\cal W_i:=\left\{w_i\in\R^{m_i}\;\big|\;|w_i|\leq\overline w_i\right\}$ for some $\overline w_i\in\R_{\geq0}$.
\end{assumption}
\begin{corollary}\label{crl:trigger}
    Consider the hybrid system $\mathcal{H}$ and suppose Assumptions 1-5 hold. We design $\Psi_i$ according to \eqref{eq:trigger}. Additionally, we design the trigger resets $\eta_i^0$ as
    \begin{equation}\label{eq:resetnon0}
        \eta_i^0(o_i):=\gamma_i\lambda_i\left(\underline{\beta}_i\max\left[|\widetilde{\epsilon}_i|-2\overline{w},0\right]\right)^2.
    \end{equation}
    The system $\mathcal{H}$ with trigger dynamics \eqref{eq:trigger} and reset \eqref{eq:resetnon0} is persistently flowing and the set $\cal A:=\{\xi: e=0\land\epsilon=0\land\eta=0\}$ is ISS if $s_i=0$ for all $i\in\mathcal{N}$. It is ISpS if, for any $i\in\mathcal{N}$, $s_i>0$.
\end{corollary}
\iffull
\begin{proof}
    Observe that $\eta_i^0$ can be upper-bounded by
    $\eta_i^0\leq\gamma_i\lambda_i\left(\underline{\beta}_i\max\left[|\widetilde{\epsilon}_i|-2\overline{w},0\right]\right)^2\leq\gamma_i\lambda_i\left(\underline{\beta}_i|\epsilon_i|\right)^2\leq\gamma_i\lambda_iW_i^2(\epsilon_i)$, where the penultimate inequality is obtained using \eqref{eq:measnetworkerror} and Assumption \ref{ass:noise}, and the final inequality by \eqref{eq:Wsandwich}. The proof then follows directly from the proof of Theorem \ref{thm:trigger}.
\end{proof}
\fi
By resetting $\eta_i$ to any positive value, it takes longer for $\eta_i$ to violate the flow condition in $\mathcal{C}_i$, i.e., to get $\eta_i=0$. Consequently, the time between consecutive transmissions may be significantly larger, especially considering that $\eta_i^0$ scales with $\widetilde{\epsilon}_i$. Thus, if the network-induced error grows rapidly when $\tau_i\in[0,\tau_\text{MIET}^i]$, fewer transmissions are required.
\begin{remark}
    As also noted in \cite{Scheres_Postoyan_Heemels_2020}, selecting the constants $s_i$ such that $s_i>0$ will lead to practical stability. However, if the noise is non-vanishing, when the estimation error is close to $0$, the inter-event times are generally close to $\tau_\text{MIET}^i$, due to the trigger resets being $0$ and $\varrho_i$ (the difference between the estimated state and the measured state), which are usually small. As demonstrated in the numerical example \iffull below \fi, choosing $s_i>0$ may increase the average inter-event times significantly when the estimation error is in the neighborhood of the origin, while not significantly impacting the asymptotic recovery of the estimated state.
\end{remark}
\iffull
\section{Numerical case study}\label{sec:numericalexample}
Consider the LTI system
\begin{equation}\label{eq:linsys}
    \Sigma:\dot{x}=Ax,\;y_i=C_ix,\;\widetilde{y}_i=y_i+w_i,\quad i\in\{1,2\},
\end{equation}
with two outputs $y_1,y_2$ and matrices
\begin{equation}
    \begin{gathered}
    A=\text{\footnotesize$\begin{bmatrix}
        0 & 2 & 0 & 0 & 0 & 1\\-2 & 0 & 1 & 0 & 0 & 0\\0 & -1 & 0 & 2 & 0 & 0\\0 & 0 & -2 & 0 & 1 & 0\\0 & 0 & 0 & -1 & 0 & 2\\-1 & 0 & 0 & 0 & -2 & 0
    \end{bmatrix}$},\\
    C_1=(1, \mathbf{0}_5),\;  C_2=(\mathbf{0}_2, 1, \mathbf{0}_3),\; C=(C_1,C_2).
    \end{gathered}
\end{equation}
This system consists of 3 interconnected marginally stable oscillators, hence the origin of \eqref{eq:linsys} is not exponentially stable (but its trajectories remain bounded). The full state of this system cannot be reconstructed from $C_1$ or $C_2$ separately, but the pair $(A,(C_1,C_2))$ is observable. We design a Luenberger-type observer, i.e.,
\begin{equation}\label{eq:linobs}
    \Omega:\dot{\chi}=A\chi+L(\psi-\yest),\;\psi_i=C_i\chi,\quad i\in\{1,2\},
\end{equation}
where $\psi:=(\psi_1,\psi_2)$. Using \eqref{eq:measnetworkerror}, the observer dynamics can be rewritten as $\dot{\chi}=A\chi+L(\psi-y-\epsilon-\what)$, resulting in the estimation and network-induced error dynamics as
\begin{equation}
    \begin{aligned}
        \dot{e}&=(A+LC)e-L\epsilon-L\widehat{w},\\
        \dot{\epsilon}_i&=C_iA\chi-C_iAx=C_iAe.
    \end{aligned}
\end{equation}
We take $W_i(\epsilon_i)=|\epsilon_i|$, resulting in $\underline{\beta}_i=\overline{\beta}_i=1$, $L_i=0$ and $H_i(e):=|C_iAe|$. We satisfy Assumption \ref{ass:V} by choosing $V(e)=e^\top Pe$, $\varrho_i(s)=s^\top Q_i s$ with $Q_i\succ0$ and solving the LMI
\begin{equation}\label{eq:LMI}
    \begin{gathered}
        \text{\footnotesize $
    \begin{bmatrix}
        \Lambda & -PL & -PL & -C^\top Q\\
        -L^\top P & \Gamma & 0 & 0\\
        -L^\top P & 0 & -\theta I & 0\\
        -QC & 0 & 0 & Q-\theta I
    \end{bmatrix}\preceq0$,}\\
    \end{gathered}
\end{equation}
where {\footnotesize $\Lambda := (A+LC)^\top P+P(A+LC) + \rho_VP + A^\top C^\top CA + C^\top QC$} and {\footnotesize $Q:=\diag(Q_1,Q_2),\;\Gamma := \diag\left(\mu_1-\gamma^2_1,\mu_2-\gamma_2^2\right)$}. One such solution is given by
\[ \text{\footnotesize $L = \begin{bmatrix}-51 & -92 & 41 & 76 & 205 & -78\\ 41 & 86 & -51 & -88 & -205 & 72\end{bmatrix}^\top,$} \]
$\mu_1=\mu_2=0.5$, $\rho_V=2$, $Q_1=Q_2=2$, $\gamma_1=\gamma_2=6.1623$ and $\theta=2.39\cdot10^{4}$. We pick $\lambda=0.7$, which fixes the constants $\tau_\text{MIET}^1=\tau_\text{MIET}^2=0.0566$.
Consequently, we design the trigger dynamics and resets using Corollary \ref{crl:trigger} as
\begin{equation}
    \begin{aligned}
        \Psi_i=&(\psi_i-\widetilde{y}_i)^\top Q_i(\psi_i-\widetilde{y}_i)-2\omega_i(\tau_i)\overline{\gamma}_i|\widetilde{\epsilon}_i|^2-\sigma_i\eta_i+s_i,\\
        \eta_i^0=&\gamma_i\lambda_i\left(\max(|\widetilde{\epsilon}_i|-2\overline{w},0)\right)^2,
    \end{aligned}
\end{equation}
with $\sigma_1=\sigma_2=0.05$. We show the influence of different choices for $s_i$ below. As noise, we use a randomly generated signal which takes values in $[-10^{-3},10^{-3}]$. A new value is randomly chosen every $10^{-4}$ seconds and held constant until the next value is chosen. Thus, the noise is a discontinuous signal.
\begin{figure}[ht!]
    \input{Figures/IET_s0.tex}
    \vspace{-2em}\caption{\footnotesize Inter-event times for event-triggered state estimation for $s_i=0$. The black dashed line represents $\tau_\text{MIET}$.}\label{fig:iets0}
    \vspace{0.25em}
\begin{tikzpicture}
\begin{axis}[%
width=8.75cm,
height=4cm,
at={(0,0)},
xmin=0,
xmax=16,
xlabel={Time [s]},
ymin=0,
ymax=3,
ylabel={Inter-event times},
axis background/.style={fill=white},
axis x line*=bottom,
axis y line*=left,
xmajorgrids,
ymajorgrids,
legend pos = north west
]
\addplot [color=black, only marks, mark=x, mark options={solid, black}]
  table[row sep=crcr]{%
0.170402735198498	0.109119987249002\\
0.336315467069826	0.165912731871328\\
0.465436166199224	0.129120699129398\\
0.613925661795821	0.148489495596597\\
0.821985708087922	0.208060046292101\\
1.38982798140219	0.567842273314271\\
1.73787173923879	0.348043757836601\\
2.08639934686838	0.348527607629588\\
2.23964124654727	0.153241899678889\\
2.38196920452938	0.14232795798211\\
2.51291215140697	0.130942946877584\\
2.64279922128187	0.129887069874909\\
2.77676230615908	0.133963084877204\\
2.9192778578001	0.142515551641021\\
3.07407064136347	0.154792783563371\\
3.24351867529715	0.16944803393368\\
3.43464016709286	0.19112149179571\\
3.65551150073241	0.220871333639555\\
3.93646376111975	0.280952260387339\\
4.46018819258404	0.523724431464283\\
4.7132136003543	0.253025407770267\\
4.92475737731741	0.211543776963108\\
5.40455284432965	0.479795467012241\\
5.84451025792819	0.439957413598536\\
6.07215992026818	0.227649662339987\\
6.32455989009705	0.252399969828878\\
6.49223160757933	0.167671717482276\\
6.63235279044919	0.140121182869858\\
6.88751126307098	0.255158472621794\\
7.17049608767197	0.282984824600988\\
7.58472621198546	0.414230124313494\\
7.86029692394533	0.275570711959865\\
8.02585482847146	0.165557904526136\\
8.26626742247364	0.240412594002176\\
8.72931977550226	0.463052353028619\\
8.85460944833103	0.125289672828774\\
9.255591244832	0.400981796500966\\
9.42448197020519	0.168890725373187\\
9.66191163933324	0.237429669128055\\
10.5437911442611	0.881879504927827\\
10.858405405177	0.314614260915882\\
11.7115763763728	0.853170971195812\\
11.8398953303265	0.128318953953706\\
12.4922593075296	0.652363977203132\\
12.5844680921374	0.0922087846077577\\
13.2077975312681	0.623329439130693\\
13.3293066329769	0.121509101708842\\
13.7082695428767	0.378962909899837\\
13.8485508814095	0.140281338532736\\
13.9559557109685	0.107404829559067\\
14.7174363286637	0.761480617695211\\
15.1909196164442	0.473483287780502\\
15.3127580196807	0.12183840323647\\
15.4718565338577	0.15909851417697\\
16.1057313054032	0.633874771545539\\
16.38801648346	0.282285178056803\\
16.6576779764424	0.269661492982419\\
16.8314001315709	0.173722155128438\\
16.946264069889	0.114863938318152\\
17.2845758671852	0.338311797296189\\
17.5732210410767	0.288645173891481\\
17.6895415538252	0.116320512748505\\
18.1029174225336	0.41337586870841\\
18.2720180004609	0.169100577927267\\
18.4626384508702	0.190620450409334\\
18.6181094674425	0.15547101657226\\
18.9550551815879	0.336945714145386\\
19.0822733532035	0.127218171615677\\
19.7660748651109	0.683801511907401\\
19.9235910628809	0.15751619776994\\
};
\addlegendentry{\footnotesize Sensor 1}
\addplot [color=red, only marks, mark=o, mark options={solid, red}]
  table[row sep=crcr]{%
0.168647183426238	0.0798265166944764\\
0.257233945428924	0.0885867620026862\\
0.363013560891954	0.10577961546303\\
0.494411635786232	0.131398074894278\\
0.662337985253436	0.167926349467204\\
0.942879239271697	0.280541254018261\\
1.64585735959289	0.702978120321194\\
1.79664283985026	0.150785480257365\\
1.92167904006013	0.125036200209873\\
2.03306195932297	0.111382919262844\\
2.14618107407232	0.113119114749348\\
2.26234221405819	0.116161139985866\\
2.38434880848378	0.122006594425589\\
2.51259553288011	0.128246724396329\\
2.64668599244687	0.134090459566764\\
2.78750865227388	0.140822659827013\\
2.93644806678715	0.148939414513265\\
3.09497777485138	0.158529708064235\\
3.26474013933524	0.169762364483858\\
3.45120709601125	0.186466956676004\\
3.66103686753892	0.209829771527669\\
3.91688088420425	0.255844016665331\\
4.24310945161849	0.326228567414241\\
4.53670890557744	0.293599453958953\\
4.79665535401702	0.259946448439574\\
4.94851225138811	0.1518568973711\\
5.14475683036656	0.196244578978448\\
5.24886894987633	0.104112119509767\\
5.33138211981757	0.082513169941242\\
5.44174560969524	0.110363489877669\\
5.56321429505631	0.121468685361068\\
5.64195324038137	0.0787389453250604\\
5.7173417147773	0.0753884743959263\\
5.81021531793033	0.0928736031530324\\
5.92398273436446	0.113767416434128\\
6.01893260628733	0.0949498719228732\\
6.11494093560373	0.0960083293163958\\
6.2275766122932	0.112635676689476\\
6.3038719258212	0.076295313528\\
6.41697474780935	0.113102821988145\\
6.56206185612937	0.145087108320022\\
6.69422269411787	0.132160837988502\\
6.83437733572154	0.140154641603671\\
6.96216487216085	0.127787536439304\\
7.14678003411728	0.184615161956433\\
7.53760372092167	0.390823686804388\\
8.08174838878991	0.544144667868244\\
8.374008752821	0.292260364031092\\
9.52064346952716	1.14663471670616\\
10.3016069806883	0.78096351116117\\
10.9020394925902	0.600432511901888\\
11.4023255231211	0.500286030530887\\
13.8692536952317	2.46692817211064\\
14.4192360572594	0.549982362027626\\
15.4054370468558	0.986200989596439\\
16.8800268551041	1.47458980824826\\
17.1431463654901	0.263119510386016\\
};
\addlegendentry{\footnotesize Sensor 2}
\addplot [dashed, color=black]
  table[row sep=crcr]{%
0	0.056690961683258\\
20	0.056690961683258\\
};
\addlegendentry{\footnotesize $\tau_\text{MIET}$}
\end{axis}
\end{tikzpicture}%
    \vspace{-1em}\caption{\footnotesize Inter-event times for event-triggered state estimation for $s_i=2\cdot10^{-4}$. The black dashed line represents $\tau_\text{MIET}$.}\label{fig:ietsnon0}
    \vspace{0.25em}\input{Figures/NormDifferentS.tex}
    \vspace{-2em}\caption{\footnotesize Norm of estimation error. The red line depicts $|e|$ for $s_i=0$ and the black line for $s_i=2\cdot10^{-4}$.}
    \label{fig:enorm}
    \vspace{-2em}
\end{figure}
In Fig. \ref{fig:iets0}, the inter-event times are depicted when $s_i=0$. Due to the trigger resets being (close to) 0, and the function $\varrho_i$ being generally small if the estimation error is close to the origin, the inter-event times are generally close to $\tau_\text{MIET}^i$. When we add space-regularization, as in \cite{Scheres_Postoyan_Heemels_2020}, the inter-event times are significantly larger. From Fig. \ref{fig:enorm}, we can see that, by picking $s_i$ sufficiently small, the ultimate bound of $|e|$ remains close to the case where $s_i=0$.
\fi
\section{Conclusions}\label{sec:conclusion}
We presented a general framework for the design of an observer that receives noisy measurement data from multiple sensor nodes over a packet-based communication network. The framework applies to nonlinear systems with disturbances and measurement noise. By design, both time-triggered and event-triggered strategies can be used to determine the transmission instants. It is shown that, in the absence of noises, the observer asymptotically reconstructs the true system state. Moreover, if noises are present, the observer error satisfies an input-to-state stability property. If event-triggered strategies are used, it is shown that Zeno behavior does not occur. \iffull Using a numerical case study, we show that favorable inter-event times can be achieved in presence of measurement noise by applying space-regularization, in which case we obtain practical convergence. By properly tuning the space-regularization parameter, the asymptotic closeness of the estimation error to the origin is not significantly impacted.\fi

\bibliographystyle{IEEEtran}
\iffull
\bibliography{references}
\fi
\iffull
\appendix
\section{Appendix}
\begin{proof}[Sketch of proof for Theorem \ref{thm:trigger}]
We consider the candidate Lyapunov function $U:\X\to\R_{\geq0}$ defined as
\begin{equation}
    U(\xi):=V(e)+\sum_{i\in\mathcal{N}}\gamma_i\phi_i(\tau_i)W^2_i(\epsilon_i)+\eta_i,
\end{equation}
for any $\xi\in\X$. From Proposition \ref{prop:phi}, $\phi_i(\tau_i)>\lambda_i$ for all $\tau_i\in\R_{\geq0}$. Additionally, $\eta_i\geq0$. Consequently, in view of \eqref{eq:Wsandwich}, \eqref{eq:Vsandwich} and the definition of attractor set $\cal A=\{\xi: e=0\land\epsilon=0\land\eta=0\}$ in Theorem \ref{thm:trigger}, there exist $\alpha_1,\alpha_2\in\mathcal{K}_\infty$ such that
\begin{equation}
    \alpha_1(|\xi|_\A)\leq U(\xi)\leq\alpha_2(|\xi|_\A)
\end{equation}
according to Proposition \ref{prop:phi}. Hence $U$ constitutes a valid Lyapunov candidate.

\paragraph{Flow analysis ($\xi\in\mathcal{C}$)}
By using Young's inequality ($xy\leq\frac{1}{2}x^2+\frac{1}{2}y^2$) and \eqref{eq:measnetworkerror}, we can upper-bound $W_i^2(\epsilon_i)$ as
$W_i^2(\epsilon_i)=W_i^2(\widetilde{\epsilon}_i+w_i-\what_i)\leq\overline{\beta}_i^2|\widetilde{\epsilon}_i+w_i-\what_i|^2\leq2\overline{\beta}_i^2\left(|\widetilde{\epsilon}_i|^2+|w_i-\what_i|^2\right)\leq\beta_iW_i^2(\widetilde{\epsilon}_i)+\delta(|w|)+\delta(|\what|)$,
where $\beta_i=2\frac{\overline{\beta}_i^2}{\underline{\beta}_i^2}$ and for some $\mathcal{K}_\infty$-function $\delta$. Then, similar to \cite{Dolk2017}, by using the facts that $2\gamma_i\phi_i(\tau_i)W_i(\epsilon_i)H_i\leq\gamma_i^2\phi_i^2(\tau_i)W_i^2(\epsilon_i) + H_i^2$ and that $\tau_i\geq\tau_\text{MIET}^i$ implies $\phi_i(\tau_i)=\lambda_i$ due to the definition of $\tau_\text{MIET}^i$ in \eqref{eq:taumiet} and Proposition \ref{prop:phi}, one can conclude that we obtain for all $\nu\in\V$ and $\tau_i\in\R_{\geq0}$, $i\in\mathcal{N}$ and almost all $\xi\in\mathcal{C}$ and any $f\in F(\xi,\nu)$,
\begin{equation}\label{eq:lyapbound}
    \begin{aligned}
        \langle\nabla U(\xi),&F(\xi,\nu)\rangle\leq-\alpha(|e|)+\overline{\zeta}(\nu)\\
        &+\textstyle{\sum}_{i\in\mathcal{N}}-\mu_iW_i^2(\epsilon_i)-\sigma_i(\eta_i)+s_i\\
        \leq&-d(U(\xi))+\overline{\zeta}(\|\nu\|)+\delta(|\what|)+s
    \end{aligned}
\end{equation}
with $d,\overline{\zeta}\in\mathcal{K}_\infty$ and $s=\sum_{i\in\mathcal{N}}s_i$.
\paragraph{Jump analysis} let $\xi\in\mathcal{D}$. We assume that a jump is generated by the $i$-th trigger.
Then, for any $g\in G_i(\xi)$,
\begin{equation}
    U(g)-U(\xi)=\eta^0_i-\eta_i-\gamma_i\phi_i(\tau_i)W_i^2(\epsilon_i)\leq0.
\end{equation}
\paragraph{Persistently flowing property}
Even though we have a hybrid system with inputs, since the state variables $\xi$ are absolutely continuous during flow and due to the fact that the flow and jumps sets $\mathcal{C}$ and $\mathcal{D}$ do not depend directly on the inputs, we can use the arguments as in \cite[Prop. 2.10]{hybridsystems} to prove nontrivial solutions exist and that all maximal solutions are complete. Note that (VC) of 2.10 holds for all $\xi\in\interior\mathcal{C}$ due to forward completeness. For all $\xi\not\in\interior(\mathcal{C}\setminus\mathcal{D})$, either $\tau_i=0$ or $\eta_i=0$ and $\tau_i\in[0,\underline{\tau}_i)$. For the former we note that $\dot\tau_i=1$ and for the latter, note that due to $\omega_i(\tau_i)$, $\Psi_i(o_i)\geq0$ for all $\tau_i\in[0,\tau_\text{MIET}^i)$. Due to $0<\underline{\tau}_i\leq\tau_\text{MIET}^i$, (VC) also holds in this case. Additionally, all maximal solutions are complete since no finite escape times are possible and due to $G(\mathcal{D})\subset\mathcal{C}\cup\mathcal{D}$. Lastly, we note that a jump (transmission) due to the $i$-th trigger can only occur after $\underline{\tau}_i>0$ time units. Take $\underline{\tau}:=\min_{i\in\mathcal{N}}\underline{\tau}_i$. Then, on any interval $[a,a+\underline{\tau})$ with $a\in\R_{\geq0}$, at most $N$ jumps (transmissions) can occur. Hence, the system is not Zeno and, therefore, persistently flowing.

Since $\mathcal{H}$ is persistently flowing, all conditions of Proposition \ref{prop:lyapunovisps} hold\footnote{Due to the definition of $\what$, see \eqref{eq:holdingfunction}, its $\Linf$-bound $\|\what\|_\infty$ can be bounded as $\|\what\|_\infty\leq\|w\|_\infty$. Hence, we can obtain the bounds of Definition \ref{def:isps} by similar arguments as \cite{Sontag1995}.}, and the set $\A$ is ISS with respect to $\nu=(v,w)$ if $s_i=0$ for all $i\in\mathcal{N}$, otherwise it is ISpS w.r.t. $\nu$.
\end{proof}
\fi
\end{document}